\documentclass{sig-alternate}
\usepackage[utf8]{inputenc}
\makeatletter
\newif\if@restonecol
\makeatother

\usepackage[ruled,vlined]{algorithm2e}
\usepackage{graphicx}
\usepackage{color}
\usepackage{amssymb}
\usepackage{amsmath}
\newtheorem{lemma}{Lemma}
\newtheorem{obs}{{\bf Observation}}
\newtheorem{theorem}{{\bf Theorem}}
\title{Formation of General Position by Asynchronous Mobile Robots}
\numberofauthors{3} 
\author{
\alignauthor
S. Bhagat\\
       \affaddr{ACM Unit}\\
       \affaddr{Indian Statistical Institute}\\
       \affaddr{Kolkata-700108}\\
       \email{\small{subhash.bhagat.math@gmail.com}}
\alignauthor
S. Gan Chaudhuri\\
       \affaddr{Department of Information Technology}\\
       \affaddr{Jadavpur University}\\
       \affaddr{Kolkata-700032}\\
       \email{\small{srutiganc@it.jusl.ac.in}}
\alignauthor K. Mukhopadhyaya\\
       \affaddr{ACM Unit}\\
       \affaddr{Indian Statistical Institute}\\
       \affaddr{Kolkata-700108}\\
       \email{\small{krishnendu@isical.ac.in}}
}

\begin{document}

\maketitle

\begin{abstract}
The traditional distributed model of autonomous, homogeneous, mobile point robots usually assumes that the robots do not create any visual obstruction for the other robots, i.e., the robots are see through. In this paper, we consider a slightly more realistic model, by
incorporating the notion of {\it obstructed visibility} (i.e., robots
are not see through) for other robots. Under the new model of visibility, a robot may not have the full view of its surroundings. Many of the existing algorithms demand that each robot should have the complete knowledge of the positions of other robots. Since, vision is the only mean of their
communication, it is required that the robots are in {\it general position} (i.e., no three robots are collinear). We consider {\it asynchronous} robots. They also do not have common {\it chirality} (or any agreement on a global coordinate system). In this paper, we present a distributed algorithm for obtaining a general position for the robots in finite time from any arbitrary configuration. The algorithm also assures collision free motion for each robot. This algorithm may also be used as a preprocessing module for many other subsequent tasks performed by the robots.
\end{abstract}


\section{Introduction}
The study of a set of autonomous mobile robots, popularly known as swarm robots or multi robot system, is an emerging 
research topic in last few decades. Swarm of robots is a set of autonomous robots that
have to organize themselves in order to execute a specific task in collaborative manner. Various problems in several directions, have been studied in the framework of swarm robots, among the others distributed computing is an important area with this swarm robots. This paper explores that direction.     
\subsection{Framework}
The traditional distributed model \cite{Peleg2005} for multi robot system, represents the mobile entities by distinct points located in the
Euclidean plane. The robots
are anonymous, indistinguishable, having no direct
means of communication. They have no common agreement in directions, orientation and unit distance. 
  Each robot has sensing capability, by {\em vision}, which enables it to determine
  the position (within its own coordinate system) of  the other robots.
The robots operate in rounds by executing {\em Look-Compute-Move} cycles. All robots may or may not be active at all rounds.
In a round, when becoming active,
    a robot gets a snapshot of its
surroundings (Look) by its sensing capability. This snapshot is used to compute 
a destination point (Compute) for this robot. Finally, it
moves towards this destination (Move). The robot either directly reaches destination or moves at-least a small distance towards the destination. The choice of
active robot in each round is decided by an
adversary. However, it is guaranteed that each robot will become active in finite time.
All robots execute the same algorithm.
The robots are oblivious, i.e., at the beginning of each cycle, they forget
their past observations and computations \cite{FlPS12}.
Depending on the activation schedule and the
duration of the cycles, three models are defined. In the {\em fully-synchronous} model, all
robots are activated simultaneously. As a result, all robots acts on same data. The {\em semi-synchronous} model is
like the fully synchronous,
except that the set of robots to be activated is chosen at random. As a result, the active robots act on same data. No assumption, is made on
timing of activation and duration of the cycles for {\em asynchronous} model. However, the time and durations are considered to be finite. 

Vision and mobility enable
    the robots to communicate and coordinate their actions by sensing their relative positions.
    Otherwise, the robots are silent and have no explicit message passing.
   These restrictions enable the robots to be deployed in extremely harsh
environments where communication is not possible, i.e an underwater
deployment or a military scenario where wired or wireless communications are
impossible or can be obstructed or erroneous. 

\subsection{Earlier works}
  Majority of the investigations\cite{EfP07, Peleg2005} on mobile robots assume that their visibility is unobstructed or full, i.e., if two robots $A$
and $B$ are located at $a$ and $b$, they
can see each other though other robots lie
in the line segment $\overline{ab}$ at that time. Very few observations on obstructed visibility (where A and B are not mutually visible if there exist other robots on the line segment $\overline{ab}$) have been made in different models; such as, (i) the robots in the one dimensional
space \cite{CoP08}; (ii) the robots with visible lights \cite{DasFPSY12,DasFPSY14} and (iii) the unit disc robot called {\it fat
robots} \cite{AgGM13, CzGP09}.

The first model studied the uniform spreading of robots  on a line \cite{CoP08}.
In the second model, each agent is provided with a local externally  visible {\em  light},
which is used as colors \cite{DasFPSY12,DasFPSY14,EfP07,FlSVY13,Peleg2005,Vi13,AFGSV14}. The
robots implicitly communicate with each other using these colors as indicators of their states.
In the third model, the robots are not points but unit discs \cite{BoKF12,CzGP09,AgGM13}) and collisions among robots are
allowed. 

Obstructed visibility have been addressed recently in \cite{AFGSV14} and \cite{AFPSV14}. In \cite{AFGSV14} the authors have proposed algorithm for robots in light model. Here, the robots starting from any arbitrary configuration form a circle which is itself an unobstructed configuration. The presence of a constant number of visible light(color) bits in each robot, implicitly help the robots in communication and storing the past configuration.  In \cite{AFGSV14}, the robots obtain a obstruction free configuration by getting as close as possible. Here, the robots do not have light bits. However, the algorithm is for semi-synchronous robots.

\subsection{Our Contribution}
In this paper, we propose algorithm to remove obstructed visibility by making of general configuration by the robots. The robots start from arbitrary distinct positions in the
plane and reach a configuration when they all see
each other. The robots are asynchronous, oblivious, having no agreement in coordinate systems. The obstructed visibility model is no doubt improves the traditional model of multi robot system by incorporating real-life like characteristic.  The problem is also a preliminary step for any
subsequent tasks which require complete visibility.

The organization of the paper is as follows: Section \ref{model}, defines the assumptions of the robot model used in this paper and presents the definitions and notations used in the algorithm. Section \ref{algo} presents an algorithm for obtaining general position by asynchronous robots. We also furnish the correctness of our algorithm in this section. Finally in section \ref{con} we conclude by providing the future directions of this work.

\section{Model and Definitions}
\label{model}
Let $\mathcal{R} = \{r_1 \ldots, r_n\}$ be a set of $n$ homogeneous robots represented by points. Each robot can sense (see) $360^o$ around itself up to an unlimited radius. 
However, they obstruct the visibility of other robots. The robots execute \emph{look-compute-move} cycle in {\it asynchronous} manner. 
They are {\it oblivious} and have {\it no direct communication} power. The movement of the robots are {\it non-rigid}, i.e., 
a robot may stop before reaching its destination. However, a robot moves at-least a minimum distance $\delta >0$ towards its destination. 
This assumption assures that a robot will reach its destination in finite time. Initially the robots are positioned in distinct locations and are stationary. 
Now we present some notations and conventions which will be used throughout the paper.
\begin{itemize}

\item \textbf{Position of a robot:} $r_i \in \cal R$ represents a location of a robot in $\cal R$ at some time, i.e., $r_i$ is a position  occupied by a robot in $\cal R$ at certain time. To denote a robot in $\cal R$ we refer by its position $r_i$.
\item \textbf{Measurement of angles:} By {\it an angle between two line segments}, if otherwise not stated, we mean the angle made by them which is less than equal to $\pi$.
\item $\mathbf{\mathbf{\cal V}(r_i):}$ For any robot $r_i$, we define the vision of $r_i$,  ${\cal V}(r_i)$, as the set of robots visible to $r_i$ (excluding $r_i$ itself). The robots in ${\cal V}(r_i)$ can also be in motion due to asynchronous scheduling.

If we sort the robots in ${\cal V}(r_i)$ by angle at $r_i$, w.r.t. $r_i$
and connect them in that order, we get  a star-shaped polygon, denoted by $STR(r_i)$. Note that $r_j \in {\cal V}(r_i) $ if and only if $r_i \in {\cal V}(r_j)$ (Figure \ref{Visionofarobot}).

\begin{figure}[h]
    \centering
   \includegraphics[scale = 1]{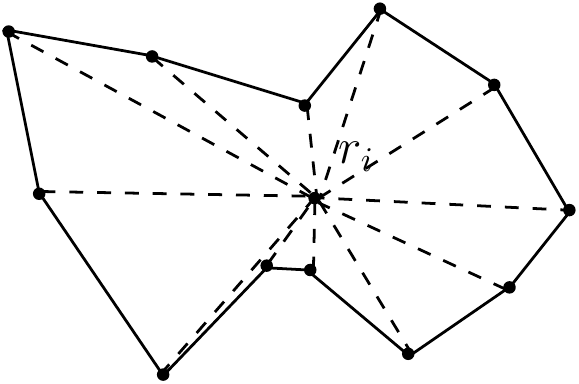}
   \caption{An example of $\mathbf{STR(r_i)}$}
   \label{Visionofarobot}
  \end{figure}

\item $\mathbf{ CR(r_i):}$ 
 This is the set of line segments joining $r_i$ to all its neighbors or all robots in ${\cal V}(r_i)$. $CR(r_i)=\{\overline{r_ir_j}: r_j \in {\cal V}(r_i)\}$ (Figure \ref{CR(r_i)}).
\begin{figure}[h]
    \centering
   \includegraphics[scale =1]{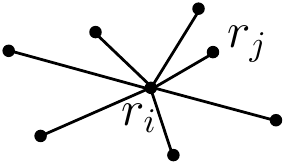} 
   \caption{An example of $\mathbf{CR(r_i)}$}
   \label{CR(r_i)}
 \end{figure}

\item $\boldsymbol{\mathcal{L}_{r_ir_j}:}$ Straight line through $r_i$ and $ r_j: r_j \in {\cal V}(r_i)$ (Figure \ref{DISP})
 \item \textbf{COL$(r_i)$:} ${COL(r_i)}$ denotes the set of robots for which $r_i$ creates visual obstructions.
\item $\mathbf{DISP(r_ir_j):}$ When a  robot $r_i$ moves to new position $\hat{r_i}$, we call   $\angle{r_ir_j\hat{r_i}}$ as the angle of displacement of $r_i$ w.r.t. $r_j$ and  denote it by $DISP(r_ir_j)$ (Figure \ref{DISP}). 

 \begin{figure}[h]
    \centering
   \includegraphics[scale = 1]{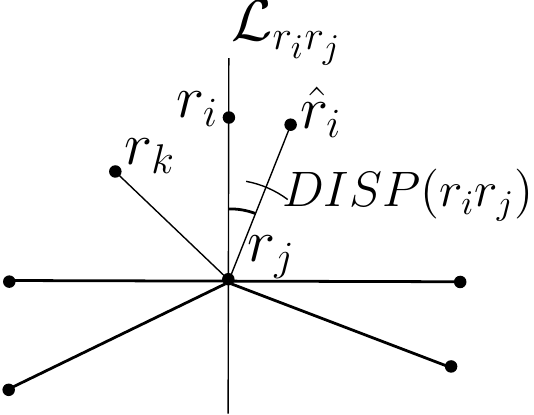}
   \caption{Examples of $\mathbf{\mathcal{L}_{r_ir_j}}$, $\mathbf{DISP(r_ir_j)=\angle{r_ir_j\hat r_i}}$, $\mathbf{COL(r_i)=\{r_l,r_m\}}$}
   \label{DISP}
  \end{figure}

\begin{figure}[h]
    \centering
   \includegraphics[scale = .75]{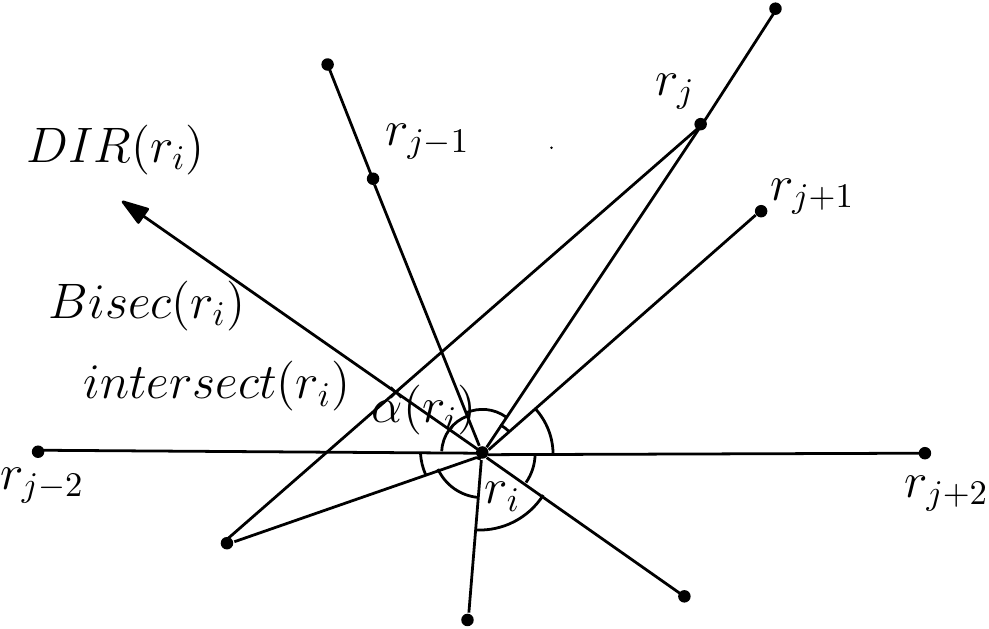}
   \caption{Examples of $\mathbf{\Gamma(r_i)}$, $\mathbf{\alpha(r_i)}$, $\mathbf{Bisec(r_i)}$, $\mathbf{intersect(r_i)}$}
   \label{alpha}
  \end{figure}
\item  $\mathbf{\Gamma(r_i:)}$ Set of angles $\angle{r_jr_ir_k}$ where $r_k$ and $r_j$ are two consecutive vertices of $STR(r_i)$ (Figure \ref{alpha}) .

\item $\boldsymbol{\alpha(r_i):}$ Maximum of $\Gamma(r_i)$ if maximum value of $\Gamma(r_i)$ is less than $\pi$ otherwise the $2^{nd}$ maximum of $\Gamma(r_i)$. The tie, if any, is broken arbitrarily (Figure \ref{alpha}).

\item $\mathbf{Bisec(r_i):}$ Bisector of $\alpha(r_i)$. Note that $Bisec(r_i)$ is a ray from $r_i$ towards the angle of consideration (Figure \ref{alpha}).

\item $\mathbf{DIR(r_i):}$ The direction of $Bisec(r_i)$. We say that $DIR(r_i)$ lies on that side of any straight line where infinite end of $DIR(r_i)$ lies (Figure \ref{alpha}).
  
\item $\mathbf{intersect(r_i):}$ We look at the intersection points of $Bisec(r_i)$ and $\mathcal{L}_{jk}$ , $\forall$ $r_j,r_k \in {\cal V}(r_i) $. The intersection point closest to $r_i$ is denoted by $intersect(r_i)$ (Figure \ref{alpha}). 
\item $\mathbf{\Gamma'(r_i)}$:  Set of angles $\angle{r_{i-1}r_jr_i}$ and $\angle{r_ir_jr_{i+1}}$, $\forall r_j \in {\cal V}(r_i)$, 
where  $r_{i-1}$ and $r_{i+1}$ are the two neighbors of $r_i$ on $STR({\cal V}(r_j))$ (Figure \ref{Gamma'}). 
\begin{figure}[h]
    \centering
   \includegraphics[scale = .75]{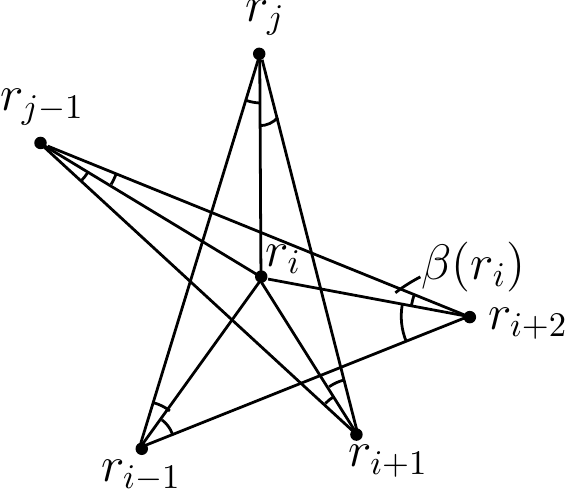}
   \caption{Examples of $\mathbf{\Gamma'(r_i)}$, $\mathbf{\beta(r_i)=\angle{r_{j-1}r_{i+2}r_i}}$}
   \label{Gamma'}
  \end{figure}
\item $\boldsymbol{\beta(r_i):}$ Minimum of $\Gamma(r_i) \cup \Gamma'(r_i)$ (Figure \ref{Gamma'}).
\item $\boldsymbol{\theta(r_i):}$ $\frac{\beta(r_i)}{n^2}$.
\item \textbf{\textit{d}($r_i$)}: Distance between $r_i$ and $intersect(r_i)$.
\item $\mathbf{D(r_i):}$ Distance between $r_i$ and the robot nearest to it.
\item $\boldsymbol{\Delta(r_i):}$ $min \{\frac{d(r_i)}{n^2},D(r_i)Sin(\theta(r_i))\}$.
\item $\boldsymbol{\hat r_i:}$ The point on $Bisec(r_i)$, $\Delta(r_i)$ distance apart from $r_i$ (Figure \ref{C(r_i)}).
\item $\boldsymbol{C(r_i):}$ The circle of radius $\Delta (r_i)$ centered at $r_i$. Note that $\hat r_i$ always 
lies on $C(r_i)$ (Figure \ref{C(r_i)}).
\item $\boldsymbol{T(C(r_i),r_j):}$ Any one of the tangential points of the tangents drawn to $C(r_i)$ from $r_j$ (Figure \ref{C(r_i)}).
\begin{figure}[h]
    \centering
   \includegraphics[scale = .75]{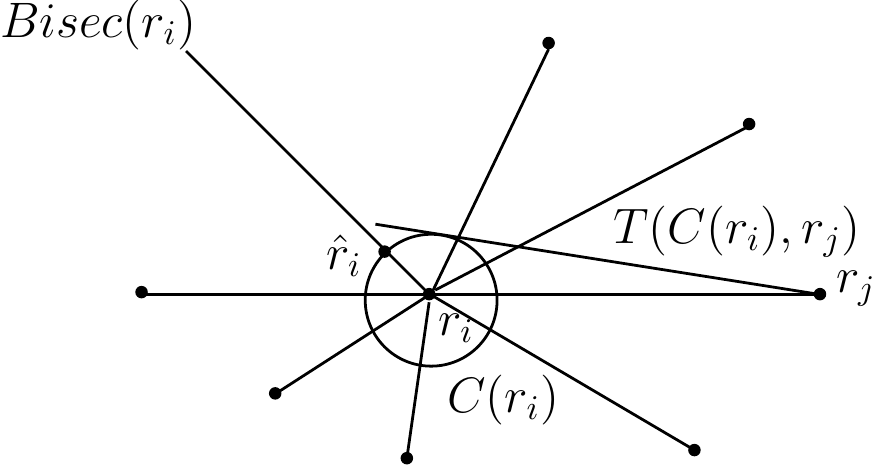}
   \caption{Examples of $\mathbf{C(r_i)}$, $\mathbf{\hat r_i}$, $\mathbf{T(C(r_i),r_j)}$}
   \label{C(r_i)}
  \end{figure}
  \end{itemize}
  
  We classify the robots in $\mathcal{R}$ depending upon their positions with respect to $\mathcal{CH(R)}$ (the convex hull of $\mathcal{R}$), as below:
  \begin{itemize}
   \item \textbf{External vertex robots ($R_{EV}$):}  A set of robots lying on the vertices of $\mathcal{CH(R)}$ . These robots do not obstruct the visibility of any robot in $\cal R$ and hence they 
   do not move during whole execution of the algorithm. 
   Note that, if $r_i$ lies outside of $STR(r_i)$ , then $r_i$ is an external vertex robot.
   \item \textbf{External edge robots ($R_{EE}$):} A set of robots lying on the edges of $\mathcal{CH(R)}$. These robots either block the visibility of external vertex robots or  other robot edge
   robots. 
   Note that, if $r_i$ lies on an edge of $STR(r_i)$, then $r_i$ is an external edge robot.
   \item \textbf{Internal robots ($R_{I}$):} A set of robots lying inside the $\mathcal{CH(R)}$. 
    Note that, if $r_i$ lies within $STR(r_i)$, $r_i$ is an internal robot.
  \end{itemize}

\section{Algorithm for Making of General Position}
\label{algo}

Consider initially robots in ${\cal R}$ are not in general position. 
Our objective is to move the robots in ${\cal R}$ in such a way that after a finite number of movements of the robots in ${\cal R}$, it will be in general position. In order to do so, our approach is to move the robots which create visual obstructions to the other robots. If a robot $r_i$ lies between two other robots, say $r_p$ and $r_q$ such that $r_i$, $r_p$ and $r_q$ are in straight line, then $r_i$ is selected for movement. The destination of $r_i$, say  $T(r_i)$, is computed in such a way that, there always exists a $r_j \in \cal R$ (where $r_j$ does not have full visibility), such that when $r_i$ moves, the cardinality of the set of visible robots of $r_j$ increases. Since, the number of robots are finite, the number of robots having partial visibility, is also finite. Our algorithm assures that at each round at-least one robot with partial visibility will have full visibility. This implies that in finite number of rounds all robots will achieve full visibility, hence, the robots will be in general 
position in finite time.  

\subsection{Computing the destinations of the robots}

A collinear middle robot is selected to move from its position. A robot finds its destination for movement using algorithm $ComputeDestination(r_i)$. A robot $r_i$, selected for moving, moves along the bisector of the minimum angle created at $r_i$ by the robots in ${\cal V}(r_i)$. The destination is chosen in such a way that $r_i$ will not block the vision of any $r_j \in {\cal V}(r_i)$, where the vision of $r_j$ was not initially blocked by $r_i$, throughout the paths towards its destination. Each movement of $r_i$ breaks at least one initial collinearity.    

\begin{algorithm}
\KwIn{$r_i \in R $ with $COL(r_i) \ne \phi$.}
\KwOut{a point on $ Bisec(r_i)$.}
\begin{enumerate}
\item Compute  $\alpha(r_i)$, $ Bisec(r_i)$, $\beta(r_i)$, $\theta(r_i)$, $D(r_i)$, \\
\vspace{.1in}
\item \textbf{Case 1:} $\beta(r_i) \ne 0$,\\ 
$\Delta(r_i) \leftarrow min \{\frac{d(r_i)}{n^2},D(r_i)Sin(\theta(r_i))\}$\\
\vspace{.1in}
\item \textbf{Case 2:} $\beta(r_i)=0$,\\
$\Delta(r_i) \leftarrow D(r_i)$

\vspace{.1in}
\item Compute the point $\hat r_i$ on $ Bisec(r_i)$, $\Delta(r_i)$ distance apart from $r_i$\;

\item return $\hat r_i$\;
\end{enumerate}

\caption{ComputeDestination()}
\end{algorithm}

\paragraph{\bf Proof of Correctness of algorithm ComputeDestination()}
Correctness of the algorithm is established by following observations, lemmas.

\begin{figure}[h]
    \centering
   \includegraphics[scale = .65]{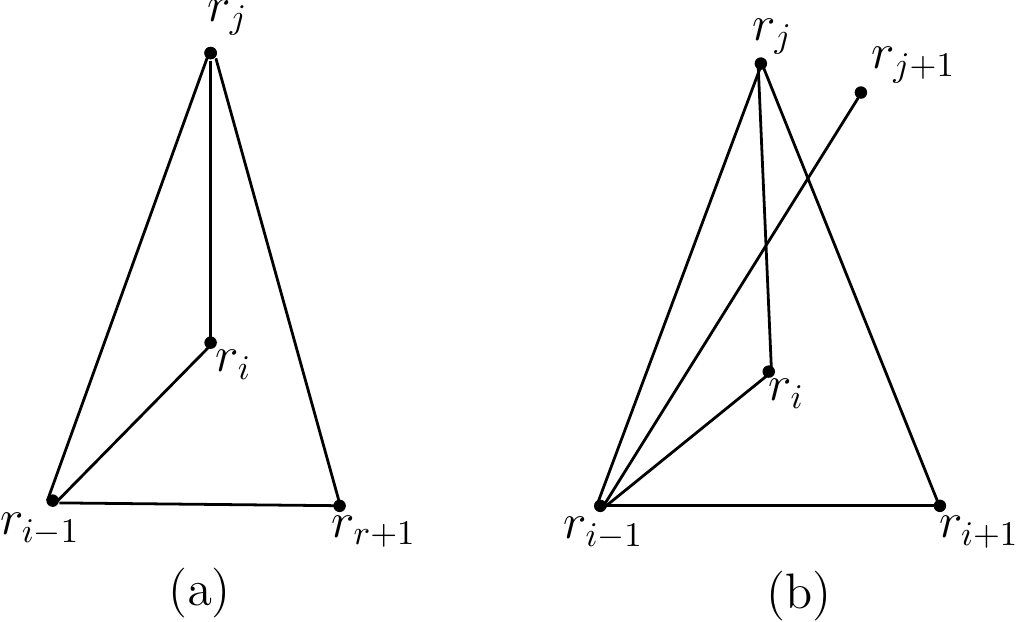}
   \caption{An example for lemma 1}
   \label{Lemma-1}
  \end{figure}
  
\begin{lemma}
\label{beta}
 $\beta(r_i)\le \frac{\pi}{3}$.
\end{lemma}
\begin{proof}
 If all the robots lie on a straight line, then $\beta(r_i) =0$. Suppose there are at least three 
 non-collinear robots. For three robots forming a triangle, $\beta(r_i)$ is maximum when the triangle is equilateral. 
 For all other cases,  consider the triangle formed by $r_i,r_j$ and $r_{i-1}$ where
 $r_j$ is any robot in ${\cal V}(r_i) $ and $r_{i-1}$ is a neighbor of $r_i$ on $STR({\cal V}(r_j))$.
If $r_j$ is also a neighbor of $r_i$ on ${\cal V}(r_{i-1}) $ (Figure \ref{Lemma-1}(a)), then 
$\angle r_ir_jr_{i-1}$ and $\angle r_ir_{i-1}r_j$ are in $\Gamma'(r_i)$ and either $\angle r_jr_ir_{i-1}$
or an angle less than it is in $\Gamma(r_i)$. On the other hand, if $r_j$ is not a neighbor of $r_i$ on ${\cal V}(r_{i-1}) $ (Figure \ref{Lemma-1}(b)), 
then instead of $\angle r_ir_{i-1}r_j$, an angle less than it, is in $\Gamma'(r_i)$. In all cases, $\beta(r_i)$ is less than the minimum of the angles of
the triangle formed by $r_i,r_j$ and $r_{i-1}$. Hence, $\beta(r_i)\le \frac{\pi}{3}$.

\end{proof}

 \begin{obs}
 \label{Max-DISP}
  Maximum value of $DISP(r_ir_j)$, denoted by $Max(DISP(r_ir_j)$, is attained when $\hat r_i$ coincides with one of the tangential points $T(C(r_i),r_j)$. 
\end{obs}

  \begin{lemma}
  \label{theta}
For any $r_i$, $DISP(r_ir_j) \le  \theta(r_i) $ $ \forall$  $r_j$.
\end{lemma}

\begin{proof}
 Let $r_j$ be a robot in ${\cal V}(r_i)$ and $r_k$ a robot closest 
 to $r_i$. By observation $\ref{Max-DISP}$, maximum values of $DISP(r_ir_j)$ and $DISP(r_ir_k)$ are attained 
 at tangential points $T(C(r_i),r_j)$ and $T(C(r_i),r_k)$  respectively. Hence, $DISP(r_ir_j)$ is less than $\frac{\pi}{2}$ for all $j$. By definition,
 \begin{align}
  \frac{\Delta(r_i)}{|\overline{r_ir_k}|} &=sin(max(DISP(r_ir_k)))\nonumber\\
 & \le sin(\theta(r_i))
  \end{align}
Again,
 \begin{equation}
  \frac{\Delta(r_i)}{|\overline{r_ir_j}|}=sin(max(DISP(r_ir_j)))  
  \end{equation}
 
Since $|\overline{r_ir_k}| < |\overline{r_ir_j}|$, from $(1)$ and $(2)$ we have,
 \begin{equation}
 sin(max(DISP(r_ir_j))) \le sin(\theta(r_i)).
 \end{equation}
\\
 $DISP(r_ir_j)$ and $\theta(r_i)$ are in  $[0,\frac{\pi}{2})$ (by lemma \ref{beta}) and $sine$ is an increasing function in $[0,\frac{\pi}{2}]$. From $(3)$ we conclude,
\begin{center}
   $DISP(r_ir_j) \le \theta(r_i)$
\end{center}
 \end{proof}

 Suppose a robot $r_i \in \mathcal{R}$ moves according to our algorithm. We claim that it will never become collinear with any two robots $r_j$ and $r_k$  in $\mathcal{R}$ where  $r_i$, $r_j$ and $r_k$ are not collinear initially. Now we state arguments to prove our claim.

\begin{obs}
 \label{right-triangle-1}
  Let $ABC$ be a right-angled triangle with $\angle ABC= \frac{\pi}{2}$. Let $D$ be a point on the side $AC$ such that $|DC| \le \frac{1}{2}|AC|$. Then,
  \begin{center}
  $\angle BDA \le 2\angle ACB$.
  \end{center}
 
 \end{obs}

 \begin{lemma}
  \label{right-triangle-2}
  Suppose $r_i$ and $r_j$ move to new positions $\hat r_i$ and $\hat r_j$  in at most one computation cycle. Let $\phi$ be the angle between $\mathcal{L}_{r_ir_j}$ and $\mathcal{L}_{\hat r_i \hat r_j}$ i.e., 
  $\phi=\angle{r_ic\hat r_i} $ where c is the intersection point between $\mathcal{L}_{r_ir_j}$ and $\mathcal{L}_{\hat r_i\hat r_j}$. Then,
  \begin{center}
   $\phi < $ 2 Max $\{\theta(r_i),\theta(r_j)\}$
  \end{center}
 \end{lemma}
 \begin{proof} 
  \begin{itemize}
If any one $r_i$ and $r_j$ moves, then lemma is trivially true. Suppose both of them move once.
  \item \textbf{Case 1:}\\
 Suppose $r_i$ and $r_j$ move synchronously. Without loss of generality, let $\Delta(r_i) \ge \Delta(r_j)$. 
  \begin{itemize}
  \item \textbf{Case 1.1:}\\ Suppose $DIR(r_i)$ and $DIR(r_j)$ lie in the opposite sides of $\mathcal{L}_{r_ir_j}$ (Figure \ref{Lemma-3-1}). In view of observation $\ref{Max-DISP}$, Max$\{\phi\}$, the maximum
    value of $\phi$, is attained  when $\mathcal{L}_{\hat r_i\hat r_j}$ is a common tangent to $C(r_i)$ and $C(r_j)$. Let $M$ be the middle point of $\overline{r_ir_j}$.
    If $C(r_i)$ is strictly larger than $C(r_j)$, $c$ is closer to $r_j$ than $r_i$. If they are equal, $c$
    coincides with $M$. Consider the right-angled triangle $\triangle{r_i\hat r_ir_j}$. By observation $\ref{right-triangle-1}$,
    
        \begin{figure}[h]
    \centering
   \includegraphics[scale = .9]{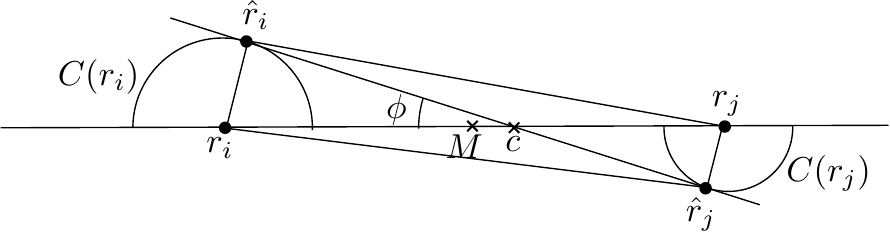}
   \caption{An example of case 1.1 for lemma 3}
   \label{Lemma-3-1}
  \end{figure}
    \begin{align}
     \phi & \le Max\{\phi\}\nonumber \\
     & \le 2 DISP(r_ir_j)\nonumber \\
     & < 2 Max \{DISP(r_ir_j)\} \nonumber \\
     &  \le  2 \theta(r_i) \nonumber
    \end{align}
    
 \item \textbf{Case 1.2:} \\If $DIR(r_i)$ and $DIR(r_j)$ lie in the same side of $\mathcal{L}_{r_ir_j}$ (Figure \ref{Lemma-3-2}), Max$\{\phi\}$ is attained when $\mathcal{L}_{\hat r_i\hat r_j}$ is a tangent to $C(r_i)$ from the point $c$ and $c$
    coincides with the closest point of $C(r_j)$ from $r_i$. Then following same argument as in case-1, we have the proof.\\
  
    \begin{figure}[h]
    \centering
   \includegraphics[scale = .9]{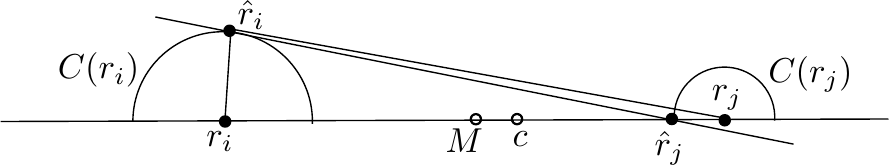}
   \caption{An example of case 1.2 for lemma 3}
   \label{Lemma-3-2}
  \end{figure}
  \end{itemize}
 \item \textbf{Case 2:}\\
 Suppose $r_i$ and $r_j$ move asynchronously. Suppose $r_i$ is moving and is at $r'_i$ when $r_j$ takes the snapshot of its  surroundings to compute
the value of $\Delta(r_j)$. Since $r_i$ has already computed the value of $\Delta(r_i)$ and computation of $\Delta$ values of $r_i$ and $r_j$ are independent, 
the proof follows from the same arguments as in case 1. In this case the value of $\Delta(r_j)$ may be different from the value in case 1.
\end{itemize}
    
  \end{proof}

 \begin{lemma}
 \label{deviation} Suppose two robots $r_i$ and $r_j$ move to $\hat r_i$ and $\hat r_j$ respectively in at most  one movement. Then\\ \\
  $Max\{DISP(r_i\hat r_j),DISP(r_j\hat r_i)\} < 2Max\{\theta(r_i),\theta(r_j)\}$.
 \end{lemma}
     \begin{figure}[h]
    \centering
   \includegraphics[scale = 1]{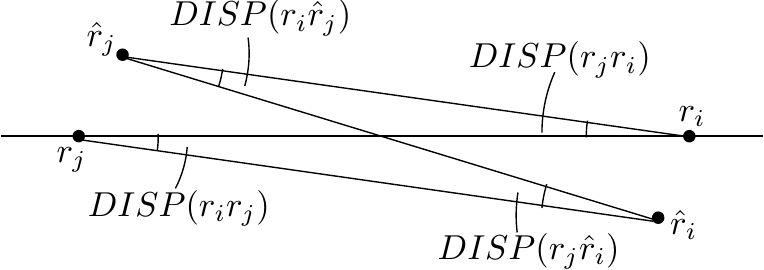}
   \caption{An example for lemma 4}
   \label{Lemma-4}
  \end{figure}
 \begin{proof}
  Follows from observation $\ref{right-triangle-1}$ and lemma $\ref{right-triangle-2}$ (Figure \ref{Lemma-4}).
 \end{proof}

\begin{lemma}
\label{lemma-5}
If $r_i,r_j$ and $r_k$ are not collinear and mutually visible to each other, then  during the whole execution of the above algorithm,  they never become collinear.
 \end{lemma}
 \begin{proof} 
 
 We have the following cases,
   \begin{itemize}
  \item \textbf{Case 1 (Only one robot moves):}\\ Without loss of generality, suppose  $r_j,r_k$ stand still and $r_i$ moves. If $DIR(r_i)$ does not intersect $\mathcal{L}_{r_jr_k}$ (Figure \ref{case--1}(a)), then the claim is trivially true.
  
  Suppose $DIR(r_i)$ intersects  $\mathcal{L}_{r_jr_k}$ (Figure \ref{case--1}(b)). Since distance traversed by $r_i$ is bounded above by $\frac{d(r_i)}{n^2}$, $r_i$ can not reach   $\mathcal{L}_{r_jr_k}$ and $r_i,r_j$ and $r_k$ will not become collinear.\\
   \begin{figure}[h]
    \centering
   \includegraphics[scale = .80]{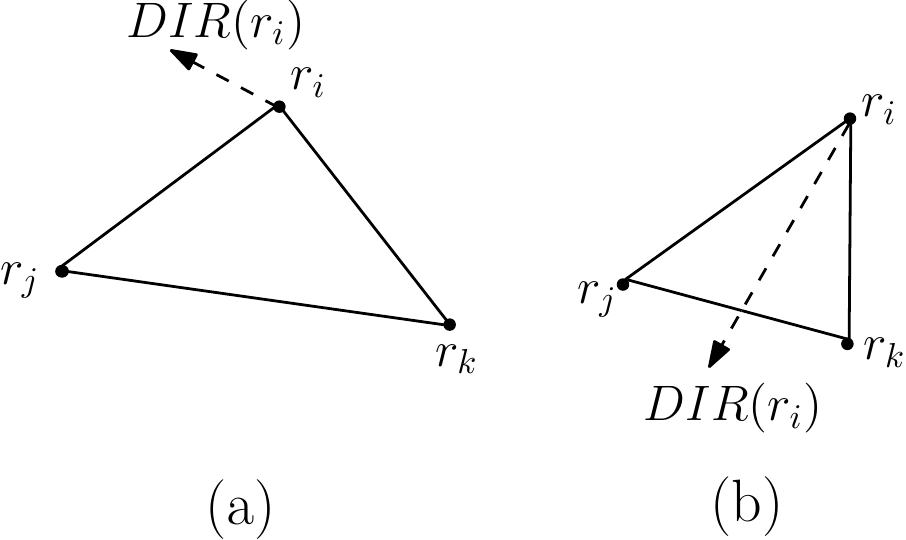}
   \caption{ An example of case 1 for lemma 5}
   \label{case--1}
 \end{figure}
 
 \item \textbf{Case 2 (Two of the robots move):}\\
 Without loss of generality, suppose $r_i$ and $r_j$ move while $r_k$ remains stationary. This case would be feasible only if $n \ge 4$. 
 
 \begin{itemize}
  \item \textbf{Case 2.1:}\\
   Suppose $r_i$ and $r_j$ move synchronously. Then by lemma $\ref{theta}$,
 
 \begin{equation}
 \label{r_i}
  DISP(r_ir_k) \le \frac{\angle{r_ir_kr_j}}{n^2}
 \end{equation}
And 
 \begin{equation}
 \label{r_j}
  DISP(r_jr_k) \le \frac{\angle{r_ir_kr_j}}{n^2}
 \end{equation}
  From equation $\ref{r_i}$ and $\ref{r_j}$
  \begin{equation}
  \label{r_ir_j}
   DISP(r_ir_k)+DISP(r_jr_k) <  \angle{r_ir_kr_j}    
  \end{equation}
The minimum value of $DISP(r_ir_k)+DISP(r_jr_k)$ for which $r_i,r_j$ and $r_k$ could become collinear is $\angle{r_ir_kr_j}$ . In view of  equation $(\ref{r_ir_j})$, we conclude that $r_i,r_j$ and $r_k$ would never become collinear. 
\item \textbf{Case 2.2:}\\
Suppose that $r_i$ is in motion and is at $\hat{ r_i}'$ when $r_j$  computes the value of $\Delta(r_j)$. If $\hat {r_i}'$  and $r_j$ lie in opposite sides of $\mathcal{L}_{r_ir_k}$ 
(Figure \ref{case2.2}(a)), 
     \begin{figure}[h]
    \centering
   \includegraphics[scale = .6]{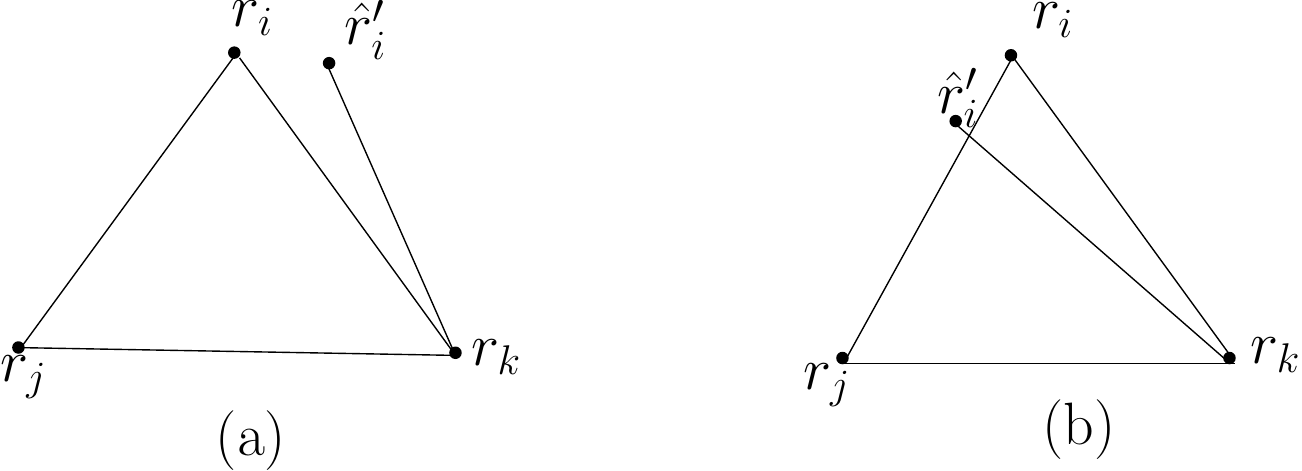}
   \caption{An example of case 2.2 for lemma 5}
   \label{case2.2}
  \end{figure}
then $$\Delta(r_j)\le \frac{1}{n^2}dist(r_j,\mathcal{L}_{r_k\hat {r_i}'})$$ which implies that $r_j$ can not reach $\mathcal{L}_{r_ir_k}$ when $r_i$ reaches its destination 
and hence the lemma. Suppose $\hat {r_i}'$  and $r_j$ lie in same side of $\mathcal{L}_{r_ir_k}$ (Figure \ref{case2.2}(b)). Then we have,
\begin{align}
 DISP(r_ir_k) & \le \frac{\angle{\hat{r_i}'r_kr_j}}{n^2}\nonumber\\
 & <  \frac{\angle{\hat{r_i}r_kr_j}}{n^2}\nonumber
\end{align}
 Lemma follows from the same arguments as used in Case $2.1$. \\
 Consider the case: suppose $r_j$ takes the snapshot at time $t$ and moves to its destination at time $t'$. In between times $t$ and $t'$, suppose $r_i$ has made at most 
 $\frac{n-1}{2}$ moves (we shall prove in case 3.2 that number of movements of any robot is bounded above by $\frac{n-1}{2}$). If $r_i$ moves towards $r_j$,  after $\frac{n-1}{2}$ moves,
 we would have $$ DISP(r_ir_j)< (1-\frac{1}{n^2})^{\frac{n-1}{2}}\angle{r_ir_kr_j} $$
 which is less than $(1-\frac{1}{n^2})\angle{r_ir_kr_j}$. Hence  equation  $(\ref{r_ir_j})$ is satisfied in this case and we have the proof of the lemma.   If $r_i$ moves away from $r_j$, then there is nothing to prove.
 \end{itemize}
 
   \item \textbf{Case 3 (All three robots move):}
  \begin{itemize}
\item \textbf{Case 3.1:}\\
Suppose $r_i,r_j$ and $r_k$ move  synchronously.
\item \textbf{Case 3.1.1:}\\
Suppose $\mathcal{L}_{\hat r_i\hat r_j}$ intersects $\mathcal{L}_{r_ir_j}$ at an angle $\phi>0$ (Figure \ref{case-3-1}). 

\begin{figure}[h]
   \centering
    \includegraphics[scale = .6]{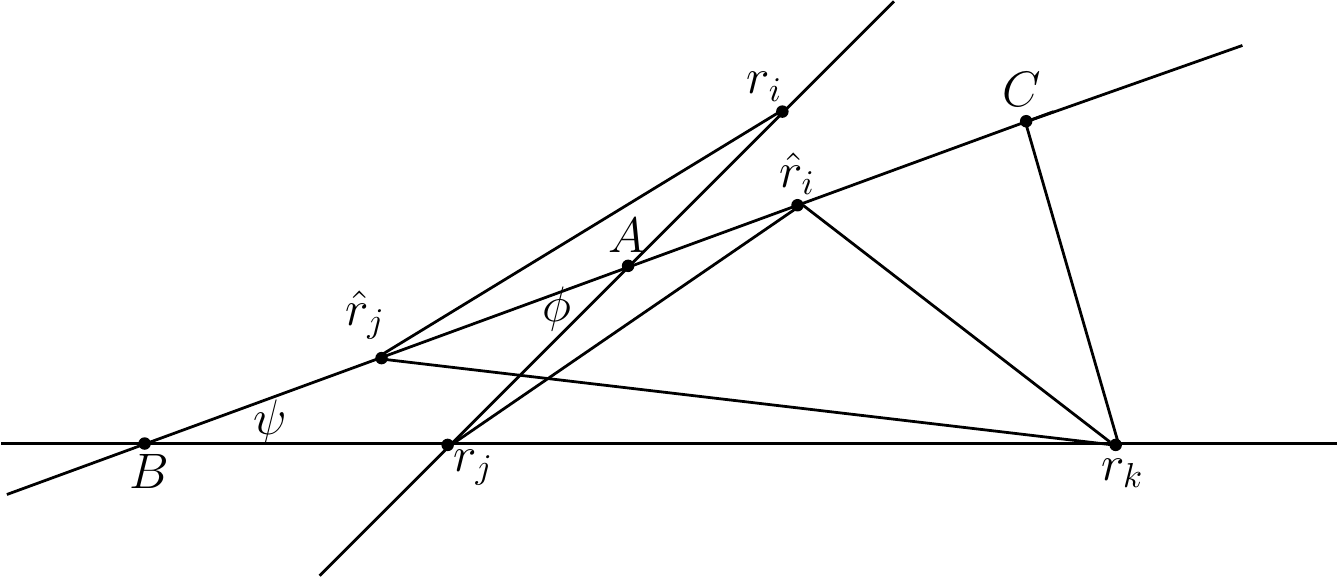}
    \caption{An example of case 3.1.1 for lemma 5}
    \label{case-3-1}
   \end{figure}
    
    By lemma $\ref{right-triangle-2}$,
    \begin{align}
    \label{key-3}
\phi & < 2 Max\{\theta(r_i),\theta(r_j)\} \nonumber\\
& \le \frac{2}{n^2}\angle {r_ir_jr_k}
\end{align}
In $\triangle{ABr_j}$, 
\begin{align}
\label{key-1}
\psi& =\angle{r_ir_jr_k}-\phi \nonumber\\ 
& > \angle{r_ir_jr_k}-\frac{2}{n^2}\angle{r_ir_jr_k}\nonumber\\
& = \frac{n-2}{n^2}\angle{r_ir_jr_k}\nonumber \\
& \ge \frac{3}{5^2}\angle{r_ir_jr_k}
\end{align}

Now $r_i$, $r_j$ and $r_k$ would be collinear only if
\begin{align}
\label{key-2}
DISP(r_kB) &=\psi   
\end{align}

From lemma  $\ref{deviation}$, 
\begin{align}
\label{key-3}
 DISP(r_kB) & < DISP(r_k\hat r_j)\nonumber\\
 &< 2 Max \{\theta(r_i),\theta(r_j)\} \nonumber \\
 & \le \frac{2}{5^2}\angle{r_ir_jr_k}
\end{align}

Equations $\ref{key-1}$, $\ref{key-2}$ and $\ref{key-3}$ imply that $r_i$, $r_j$ and $r_k$ do not become collinear. 
\item \textbf{Case 3.1.2:}\\
Suppose $\mathcal{L}_{\hat r_i\hat r_j}$ and $\mathcal{L}_{r_ir_j}$ are parallel i.e.,  $\phi=0$ which implies that $\psi=\angle{r_ir_jr_k}$ (Figure \ref{Lemma-5-312}). Let $Bisec(r_k)$ intersect
 $\mathcal{L}_{r_ir_j}$ at $P$ and $|\overline{r_kP}|=l$. Since $\Delta(r_k) \le |\overline{r_ir_j}|sin(\frac{\angle{r_ir_jr_k}}{n^2})$ and $n\ge 5$,
     \begin{figure}[h]
    \centering
   \includegraphics[scale = 1]{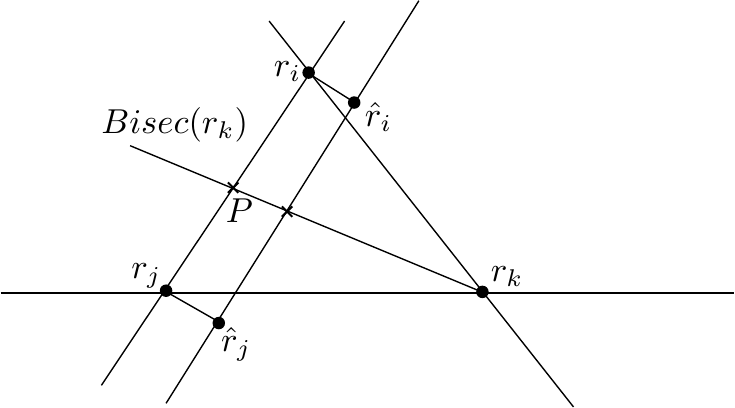}
   \caption{An example of case 3.1.2 for lemma 5}
   \label{Lemma-5-312}
  \end{figure}
\begin{align}
\label{eq-11}
 l-\Delta(r_k) & \ge |\overline{r_ir_j}|sin(\angle{r_ir_jr_k})- \Delta(r_k)\nonumber\\
 & \ge |\overline{r_ir_j}|sin(\angle{r_ir_jr_k})-|\overline{r_ir_j}|sin(\frac{\angle {r_ir_jr_k}}{n^2})\nonumber \\
 & \ge |\overline{r_ir_j}|(sin(\angle r_ir_jr_k)-sin(\frac{\angle {r_ir_jr_k}}{5^2}))\nonumber\\
 & > |\overline{r_ir_j}|sin(\frac{\angle {r_ir_jr_k}}{5^2}) 
\end{align}

$\Delta(r_i)$ and $\Delta(r_j)$ are bounded above by $|\overline{r_ir_j}|sin$ $(\frac{\angle {r_ir_jr_k}}{5^2})$. Hence by equation $(\ref{eq-11})$, 
$r_i$ and $r_j$ and $r_k$ do not become collinear.

 \end{itemize}

\item \textbf{Case 3.2:}

Suppose $r_i$, $r_j$ and $r_k$ move  asynchronously. The main problem in this case is the following scenario:  
suppose $r_j$ or $r_k$  takes the snapshot at time $t_j$ or $t_k$ respectively and starts  moving to its computed destination at time $t'_j$ or $t'_k$ respectively. Suppose the configuration has
been changed in between the times due to the movements of the other robots. Then the corresponding $\Delta$ value of $r_j$ or $r_k$  is not consistent w.r.t. the current configuration. We have to show that this would not 
create any problem for our algorithm.
The main idea of proof in this case is that we have to estimate the maximum amount of inclination of $\mathcal{L}_{r_ir_j}$ towards $r_k$ between the times $r_j$ or $r_k$ takes
the snapshot of surroundings and it
reaches the destination. So, in the following proofs we only consider the scenarios (as in the case 3.1.1. and case 3.1.2) in which there are possibilities of maximum reduction in the $\angle{r_ir_jr_k}$, which depicts the inclination 
of $\mathcal{L}_{r_ir_j}$ towards $r_k$. Note that the inclination of $\mathcal{L}_{r_ir_j}$ towards $r_k$ is maximum when both $r_i$ and $r_j$ move synchronously. So, we only prove the case when $r_k$ holds the old
value of $\Delta$. 

\item \textbf{Case 3.2.1}\\
Suppose $r_k$ holds the old value of $\Delta$ w.r.t. to the current configuration.
Suppose $r_i$ and $r_j$ are at 
$r^0_i$ and  $r^0_j$ respectively when $r_k$ takes the snapshot at time $t_k$. Suppose till $t'_k$, $r_i$ and $r_j$ move $x$ and $x'$ times respectively. Note that initially $r_i$ and $r_j$ can
be collinear with $n-1$ robots and to remove these collinearity they  have to move at most $\frac{n-1}{2}$ times if they do not create any new collinearity (this bound is obtained by 
considering the degenerate case i.e., when all the robots are collinear initially).\\
First we prove that
$x$ and $x'$ are bounded above by $\frac{n-1}{2}$. To prove this we show that $r_i$ and $r_j$ do not create any new collinearity while moving. We prove this for arbitrary robots. 
Suppose some robot $r_s$, while moving, creates a new collinearity with $r_l$ and $r_m$ for the first time during the 
execution of our algorithm (Figure \ref{case 3.2}). 
  \begin{figure}[h]
    \centering
   \includegraphics[scale = .8]{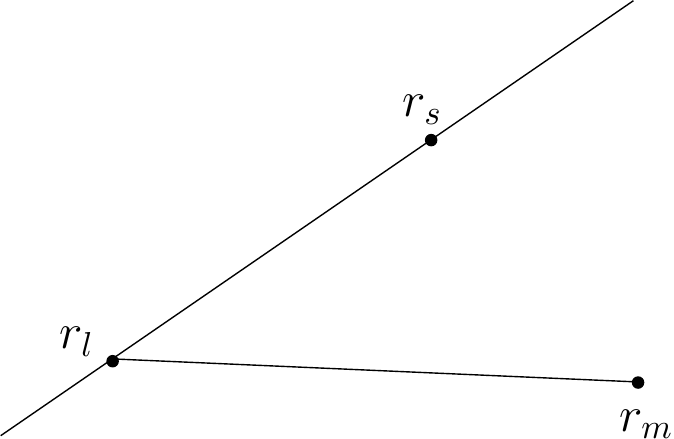}
   \caption{An example of case 3.2.1 for lemma 5}
   \label{case 3.2}
  \end{figure}
Then either one of $r_l$ and $r_m$ or both have  $\Delta$ values w.r.t. old configurations. As stated earlier we only prove the case in which only one robot, say $r_m$, has old $\Delta$
value. $r_m$ computes $\Delta(r_m)$ at the time $t_m$ i.e., $$\Delta(r_m)\le \frac{1}{n^2}\angle{r_sr_lr_m}.$$ Suppose  $r_m$ does not move 
till time $t'_m$. The number of times $r_s$ and $r_l$ move to break the initial collinearities before time $t'_m$ is upper bounded by $\frac{n-1}{2}$. 
 $r_m$ would become
collinear with $r_s$ and $r_l$ when $\mathcal{L}_{r_sr_l}$ would be inclined enough towards $r_m$  so that by moving a $\Delta(r_m)$ amount it would reach this straight line. 
We try to estimate the inclination of $\mathcal{L}_{r_sr_l}$ towards $r_m$ (which is depicted by the angle $\psi$ as in 
the case 3.1.1. and by the displacement of $\mathcal{L}_{r_sr_l}$ towards $r_m$ as in the case 3.1.2.) after  $\frac{n-1}{2}$ number of movements of $r_s$ and $r_l$ (note that we have consider the over estimated value of the number of movements of $r_s$ and $r_l$). 
As computed in the case 3.1.1, after first movement,  $$\psi > (1- \frac{1}{n^2}) \angle{r_sr_lr_m}$$ and $\angle{r_sr_lr_m}$ will become at most $(1+\frac{1}{n^2})\angle{r_sr_lr_m}$. 
By the same repeated arguments, we can say that after $d$ movements $$\psi>(1-\frac{1}{n^2})^d\angle{r_sr_lr_m}$$ which is strictly greater than $\frac{1}{n^2}\angle{r_sr_lr_m}$ for $d\le\frac{n-1}{2}$. 
 This contradicts the fact that $r_s$ creates collinearity with $r_l$ and $r_m$.  
 For the scenario same as the case 3.1.2., we have,\\
 \begin{align}
 \label{dis}
 |\overline{r_lr_m}|sin(\angle{r_sr_lr_m})-\frac{n-1}{2}|\overline{r_lr_m}|sin(\frac{\angle{r_sr_lr_m}}{n^2}) & > \nonumber \\ 
 |\overline{r_lr_m}|sin(\frac{\angle{r_sr_lr_m}}{n^2})
 \end{align}
 
This also  contradicts the fact that $r_s$ creates collinearity with $r_l$ and $r_m$. 
Hence, we conclude that $r_s$ would not become collinear with $r_l$ and $r_m$.\\
In the above proof, we replace $r_s$, $r_l$ and $r_m$ by $r_i$, $r_j$ and $r_k$ respectively to conclude that $r_i$ would not become collinear with $r_j$ and $r_k$ during the whole execution of our algorithm.

\end{itemize}
 \end{proof}
 
\begin{lemma}
\label{lemma6}
Consider any two robots $r_i$ and $r_j$. $r_i$ does not cross $Bisec(r_j)$. 
\end{lemma}
     \begin{figure}[h]
   \centering
    \includegraphics[scale = .7]{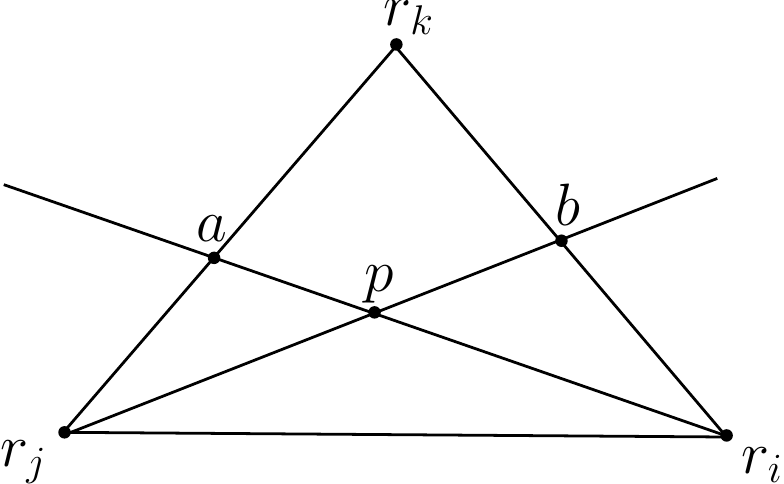}
    \caption{An example for lemma 6}
    \label{lemma-6}
   \end{figure}
\begin{proof}
 If $Bisec(r_i)$ and $Bisec(r_j)$ do not intersect, then there is nothing to prove.  Suppose $Bisec(r_i)$ and $Bisec(r_j)$  intersect at a point $p$ (Figure \ref{lemma-6}). If at least one of $intersect(r_i)$ and $intersect(r_j)$ is closer to $r_i$ and
 $r_j$ respectively than $p$, then we are done. Else $\alpha(r_i)$ and $\alpha(r_j)$ are angle of same triangle $\triangle{r_ir_jr_k}$ for some $r_k \in \cal R$ i.e, $\alpha(r_i)=\angle{r_kr_ir_j}$ and $\alpha(r_i)=\angle{r_kr_jr_i}$. 
 In $\triangle{r_ir_jr_k}$, let $Bisec(r_i)$ and $Bisec(r_j)$ intersect $\overline{r_jr_k}$ and $\overline{r_ir_k}$ at $a$ and $b$ respectively. Here $n > 5$.\\ \\
 In $\triangle{ar_jp}$,
 \begin{align}
  \label{key-6}
  |\overline{ap}| & = sin(\frac{\angle{r_kr_jr_i}}{2})\frac{|\overline{r_ja}|}{sin(\angle{apr_j})}
 \end{align}
 In $\triangle{pr_ir_j}$,
\begin{align}
 \label{key-7}
 |\overline{pr_i}| & = sin(\frac{\angle{r_kr_jr_i}}{2})\frac{|\overline{r_ir_j}|}{sin(\angle{\pi-apr_j})} \nonumber \\
 & =sin(\frac{\angle{r_kr_jr_i}}{2})\frac{|\overline{r_ir_j}|}{sin(\angle{apr_j})}
 \end{align}
From equation $\ref{key-6}$ and $\ref{key-7}$,
\begin{align}
 \frac{|\overline{ap}|}{|\overline{pr_i}|}= \frac{|\overline{r_ja}|}{ |\overline{r_ir_j}|}
\end{align}
Since $|\overline{r_ja}|< |\overline{r_ir_j}|$, $|\overline{ap}|<|\overline{pr_i}|$ which implies, 
\begin{center}
$\Delta(r_i)<\frac{|\overline{r_ia}|}{5^2}$\\
\hspace{.32in}< $|\overline{pr_i}|$.
\end{center}
Hence $r_i$ can not 
cross $Bisec(r_j)$. Similarly, $r_j$  can not cross $Bisec(r_i)$.
\end{proof}

 \begin{lemma}
 \label{lemma7}
  
 Suppose, for any robot $r_i \in \cal R$, $r_k \notin {\cal V}(r_i)$. 
 Then  during the whole execution of the algorithm $r_i$ will not block the vision between  $r_j$ and  $r_k$ where $r_j \in {\cal V}(r_k)$. 
 \end{lemma}

      \begin{figure}[h]
   \centering
    \includegraphics[scale = .7]{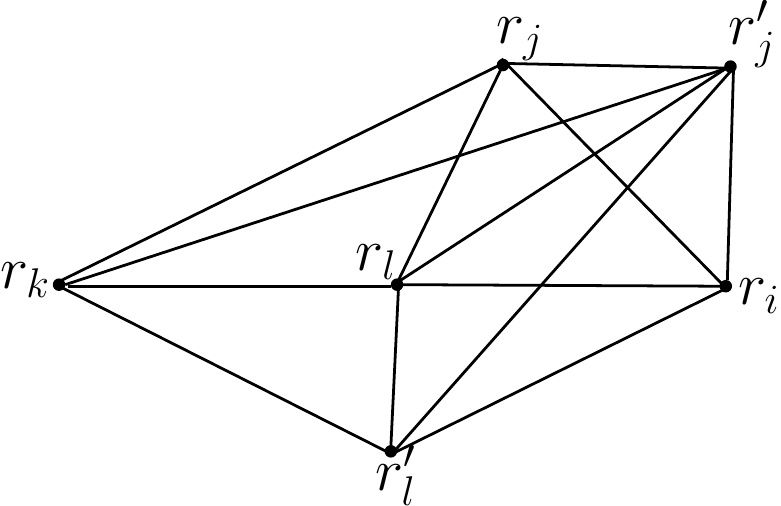}
    \caption{An example for lemma 7}
    \label{lemma-7}
   \end{figure}
 \begin{proof}
 Let $r_j \in {\cal V}(r_i) \cap  {\cal V}(r_k) $. Suppose $r_l$ be the nearest robot of $r_i$ such that $r_k$ lie on  $\mathcal{L}_{r_ir_l}$ (Figure \ref{lemma-7}). If $Bisec(r_i)$ does not intersect $\overline{r_jr_l}$, there is no possibility that $r_i$ 
will block the vision between $r_j$ and $r_k$. Let $Bisec(r_i)$ intersect $\overline{r_jr_l}$. Then  $r_j$ is one of the immediate neighbor of $r_l$ on $STR({\cal V}(r_i))$.
Let $r'_j$ and $r'_l$ be the other immediate neighbors of $r_j$ and $r_l$ respectively on $STR({\cal V}(r_i))$.  First we prove that $r_i$ will always lie on the same side of 
$\mathcal{L}_{r_jr_l}$ as it is  initially even if $r_i$, $r_j$, $r_k$ and $r_j$ move. By lemma $\ref{lemma6}$ and the observation that the movements of  $r_i$, $r_j$, $r_l$ are bounded by the edges and  chords
of the polygon formed by $\{r_j,r_l,r'_l,r_i,r'_j\}$, we conclude $r_i$ never crosses the line $\mathcal{L}_{r_jr_l}$. To block the vision between $r_k$ and $r_j$, $r_i$
has to move on the line segment $\overline{r_kr_j}$. Since $r_i$ and line segment $\overline{r_jr_k}$ lies
 on different sides of $\mathcal{L}_{r_jr_l}$, $r_i$ will never block the vision between $r_k$ and $r_j$.
  Let $r_j \notin {\cal V}(r_i).$ Then there is a robot $r_m$ which creates visual obstruction between $r_i$ and $r_j$. Now the movement of $r_i$ is bounded by the line 
 $\mathcal{L}_{r_lr_m}$ and hence the lemma.  
 \end{proof}

\begin{lemma}
\label{lemma-8}
 If at any time $t$, $r_j \in {\cal V}(r_i)$, then at $t' (>t)$, $r_j \in {\cal V}(r_i)$  even if $r_i$ changes its position. 
\end{lemma}
\begin{proof}
The proof is immediate from \ref{lemma-5} and \ref{lemma7}.
 \end{proof}

\begin{lemma}
\label{lemma-9}
Cardinality of ${\cal V}(r_i)$ is strictly increasing.
\end{lemma}
\begin{proof}
 Lemma $\ref{lemma-5}$, $\ref{lemma7}$ and $\ref{lemma-8}$ imply the proof.
\end{proof}

\begin{lemma}
\label{lemma-10}
 There exist at least two robots $r_j, r_k \in \cal R$ for which ${\cal V}(r_j)$ and ${\cal V}(r_j)$ increase whenever $r_i$  changes its position.
\end{lemma}
\begin{proof}
$r_i$ moves whenever $r_i$ is  collinear with at least one pair of robots, ($r_j$, $r_k$), and $r_i$ lies in between those robots. 
If $r_j$ and $r_k$ do not move then ${\cal V}(r_j)$ and ${\cal V}(r_k)$ increase whenever $r_i$ moves because no robot can reach $\overline{r_jr_k}$ due to the 
facts stated in lemma $\ref{lemma-5}$ and $\ref{lemma7}$. When either $r_j$ or $r_k$ or both $r_i$ and $r_k$ moves, one member of $COL(r_j)$ and 
one member of $COL(r_k)$ can see each other. Hence the lemma.
\end{proof}

\subsection{Moving the robots to obtain general position}

Next we will discuss the algorithm $MakeGenaralPosition()$, by which the robots in $\cal R$ move to obtain full visibility. The robots in $R_I$ which create obstacle to other robots and the robots in $R_{EE}$ are eligible for movement by this algorithm. The robots compute destinations using $ComputeDestination()$ and move towards it. The robots keep on executing the algorithm till there exist no three collinear robots in $\cal R$.

\begin{algorithm}
\KwIn{$\cal R$, a set of robots with their positions.}
\KwOut{$\cal \hat R$,  which is in general position.}
\While{$r_i\in R_{EE}$ $\vee$ ($r_i \in R_I \wedge COL(r_i) \ne \phi$)}{
\begin{enumerate}
\item $T(r_i) \leftarrow ComputeDestination(r_i)$\;
\item Move to $T(r_i)$\;
\item Compute $COL(r_i)$\;
\end{enumerate}
}

\caption{MakeGenaralPosition()}
\end{algorithm}
 
\paragraph{\bf Proof of Correctness of algorithm MakeGenaralPosition()}
The algorithm assures that the robot will form general position in finite number of movements.
The termination of the algorithm is established by following observation and lemmas.

\begin{obs}
  $ComputeDestination$ is not executed by a robot  $r_l\in R$ if  $r_l \in R_{EV}\vee (R_I\wedge COL(r_l)= \phi)$.
 \end{obs}
 
\begin{lemma}
\label{finitecolphi}
 $COL(r_i)$ will be $\phi$ in finite time. 
\end{lemma}
\begin{proof}
In the initial configuration the number of robots in $COL(r_i)$ is upper bounded by $n-1$. During the whole  execution of our algorithm no new collinearity is created and for each iteration cardinality of $COL(r_i)$ is reduced by
at least two. Hence after at most $\frac{n-1}{2}$ number of iterations of the while loop in the above algorithm, $COL(r_i)$ will become null. 
\end{proof}

\begin{lemma}
 $\forall r_i, {\cal V}(r_i)$ will be $(n-1)$ in finite number of execution of the cycle. 
\end{lemma}
 \begin{proof}
 Let $\eta= |\bigcup_{i=1}^n{\cal V}(r_i)|$. The algorithm for a  robot $r_i$ terminates whenever $|{\cal V}(r_i)|$ reaches the value $n-1$. Hence the algorithm for all robots 
 terminates when $\eta=\frac{n(n-1)}{2}$ which is a finite integer. By lemma $\ref{lemma-9}$ and $\ref{lemma-10}$ the value of $\eta$ increases whenever any robot moves. Hence
 after finite number of execution cycles $\eta$ reaches its maximum value $\frac{n(n-1)}{2}$.
 \end{proof}

From the above results, we can conclude the following theorem:

\begin{theorem}
 A set of asynchronous, oblivious robots (initially not in general position) without agreement in common chirality, can form general position in finite time. 
\end{theorem}

\section{Conclusion}
\label{con}
In this paper we have presented an algorithm for obtaining general position by a set of autonomous, homogeneous, oblivious, asynchronous robots having no common chirality. The algorithm assures the robots to have collision free movements. Another important feature of our algorithm is that the convex hull made by the robots in initial position, remains intact both in location and size. In other words, the robots do no go out side the convex hull formed by them. This feature can help in many subsequent pattern formations which require to maintain the location and size and of the pattern.    

Once the robots obtain general position, the next job could be to form any pattern maintaining the general position. Most of the existing pattern formation algorithms have assumed that the robots are see through. Thus, designing algorithms for forming patterns by maintaining general position of the robots, may be a direct extension of this work.   


\end{document}